\documentclass[preprint]{imsart}

\usepackage{amsthm,amsmath,natbib}
\usepackage{amssymb}
\usepackage{enumitem}
\usepackage{graphicx}
\RequirePackage[colorlinks,citecolor=blue,urlcolor=blue]{hyperref}
\usepackage{bm}

\startlocaldefs
\newcommand{\mc}{\mathcal}
\newcommand{\mb}{\mathbb}
\newtheorem{theorem}{Theorem}
\newtheorem{prop}[theorem]{Proposition}

\theoremstyle{definition}

\newtheorem*{remark}{Remark}
\newtheorem*{remarks}{Remarks}
\endlocaldefs

\newenvironment{proof-of-lemma}[1]{\noindent{\bf Proof of Lemma #1}\hspace*{1em}}{\qed\bigskip}

\usepackage{caption}
\usepackage{subcaption}

\begin{document}

\bibliographystyle{plainnat}

\begin{frontmatter}


\title{Bouncy Hybrid Sampler as a Unifying Device}
\runtitle{Bouncy Hybrid Sampler}

 \author{Jelena Markovic\corref{}\ead[label=e2]{jelenam@stanford.edu}}
 \and
 \author{Amir Sepehri\corref{}\ead[label=e1]{asepehri@stanford.edu}}
 
 \address{Department of Statistics \\ Sequoia Hall, Stanford, CA 94305, USA\\ \printead{e2}}
 \affiliation{Stanford University}
\runauthor{J.~Markovic and A.~Sepehri}

\begin{abstract}
This work introduces a class of rejection-free Markov chain Monte Carlo (MCMC) samplers, named the Bouncy Hybrid Sampler, which unifies several existing methods from the literature. Examples include the Bouncy Particle Sampler of \cite{PhysRevE.85.026703, bouchard2015bouncy} 
 	 and the Hamiltonian MCMC. Following the introduced general framework, we derive a new sampler called the Quadratic Bouncy Hybrid Sampler. We apply this novel sampler to the problem of sampling from a truncated Gaussian distribution.
\end{abstract}



\end{frontmatter}


\section{Introduction} \label{sec:intro} 

Markov chain Monte Carlo (MCMC) methods are central tools to sample complex distributions in many applications in sciences and engineering. There are already numerous MCMC samplers developed in the literature. This paper addresses the problem of connecting some of the most recently developed Markov chains via a unifying framework.

Traditional Markov chains, including Metropolis-Hastings and Hamiltonian Markov chain, are reversible by construction, meaning that they behave similarly when considered ``forward in time'' or ``backward in time.'' Markov chains have been extensively studied in the literature \citep{duane1987hybrid, neal2011mcmc}.
Non-reversible Markov chains increased in popularity as it became known that they can converge faster to a target distribution than the reversible ones. The original examples of non-reversible chains were constructed by ``lifting'' the reversible ones, i.e.~by splitting each state into several states \citep{diaconis2000analysis, chen1999lifting}. Beyond these constructions by lifting, the non-reversible chains can be harder to construct.

\cite{PhysRevE.85.026703} introduced the infinitesimal Metropolis-Hastings filter and combined it with the lifting framework to construct a non-reversible rejection-free continuous time Markov process to sample from a density function on $\mb{R}^d$. \cite{bouchard2015bouncy} analyzed this method, proving that the target distribution is the invariant measure of the corresponding Markov process. They named the procedure ``Bouncy Particle Sampler'' (BPS) and considered various aspects of implementation. In BPS, the ``particle'' moves along a straight line, applying the infinitesimal Metropolis-Hasting filter when facing an energy barrier. In the rejection event of an infinitesimal Metropolis-Hastings filter, the move is not rejected; instead, the particle bounces against the energy barrier. Furthermore, they showed the chain is irreducible with \cite{deligiannidis2017exponential} proving it is geometrically ergodic. Further modifications and applications of the BPS are already presented in many works, including \cite{wu2017generalized, pakman2016stochastic, pakman2017binary, sepehri2017non}.

Another similar piecewise deterministic and non-reversible process with a particle moving along linear lines is the Zig-Zag process introduced in \cite{bierkens2016zig}. The difference between the BPS and the Zig-Zag Markov chain is that in a Zig-Zag process, the velocity changes only along a single coordinate at every trajectory switch.

To unify these recently developed Markov chains, we introduce a class of rejection-free Markov chain Monte Carlo samplers, named the Bouncy Hybrid Sampler (BHS). The BHS is written relying on piecewise deterministic Markov processes (PDMP).
There is already some literature pointing out the recently developed continuous time samplers, including the BPS and a class of sequential MCMC algorithms, are special instances of piecewise deterministic Markov processes \citep{fearnhead2016piecewise}. Writing a sampler in PDMP language is useful as it allows using the classical results from PDMP literature.

In the BHS, the particle moves along a, not necessarily linear, trajectory with a time-dependent speed. 
The position and velocity functions are govern by a system of differential equations. The change in velocity is depicted via a chosen function of position. To ensure the sampler converges to the right distribution, the moving time on a single trajectory is determined via a Poisson process whose rate depends on the function of choice. After moving for a random time sampled as the first arrival time of the Poisson process, the particle switches the trajectories.

The proposed BHS family presents an infinite class of samplers, where BPS, Zig-Zag and Hamiltonian MC are the special instances.

\subsection{Outline}

Section \ref{sec:background} provides the necessary background on Markov processes and more specific PDMP. Section \ref{sec:hmc} introduces a novel family of so called the \underline{\textit{Bouncy Hybrid}} Markov chain Monte Carlo samplers. The computational aspects of the algorithm are detailed in Section \ref{sec:algorithm}.
Section \ref{sec:tmvg} illustrates a new sampler derived from this class called the \underline{\textit{Quadratic Bouncy Hybrid}} MCMC, applied to sampling from a truncated normal distribution. 
Section \ref{sec:extensions} provides further modifications and generalizations of the proposed BHS family to create an even bigger family of samplers. Section \ref{sec:zig:zag} presents another application of piecewise deterministic Markov chains to create the \underline{\textit{Coordinate Bouncy Hybrid}} MCMC.


\section{Background} \label{sec:background}

\subsection{Continuous Time Markov Processes}

This section provides a very brief introduction to concepts and facts from the theory of Markov processes, which will be used in later sections; for a textbook length treatment see, for example, \cite{kolokoltsov2011markov}.
A stochastic process $\left\{Z_t \mid t\ge 0\right\}$ on a measurable space $(\mathcal{Z},\mc{B})$ is a collection of random variables assuming values in $\mathcal{Z}$. Formally, $\left\{Z_t \mid t\ge 0\right\}$ is defined on the probability space $(\Omega, \mc{F},\mb{P})$, where $\Omega = \left\{f:[0,\infty) \rightarrow \mathcal{Z}\right\}$, $\mc{F}$ is the $\sigma$-algebra generated by the sets ${f\in \Omega \mid f(t) \in B}$ for $t\ge 0$ and $B \in \mc{B}$, and $\mb{P}$ is the probability measure corresponding to the law of $\left\{Z_t \mid t\ge 0\right\}$. With some abuse of notation, $Z_t$ will be used for the stochastic process $\left\{Z_t \mid t\ge 0\right\}$.

A process $Z_t$ is a \textit{Markov process} if $\mb{P}\left\{Z_t \in B \mid Z_s; s \le t_0\right\} = \mb{P}\left\{Z_t \in B \mid Z_{t_0}\right\}$ for all $t_0 < t$ and $B\in \mc{B}$.
In other words, if the probabilistic dependence of the future on the past is through the present value. A Markov process is called \textit{homogeneous} if $\mb{P}\left\{Z_{t+h} \in B \mid Z_t\right\} = \mb{P}\left\{Z_h \in B \mid Z_0\right\}$ for all $ t,h\geq 0$ and $B\in \mc{B}$.
Let $p(t,z,B) = \mb{P}\left\{Z_t \in B \mid Z_0 = z\right\}$ be the \textit{transition kernel} associated with the process $Z_t$. A probability measure $\mu$ is called an \textit{invariant measure} for $Z_t$ if
\begin{equation*}
	\int p(t,z,B) \mu(dz) = \mu(B)  \qquad \forall B\in \mc{B}\text{ and }\forall t \in [0,\infty).
\end{equation*}

To every homogeneous Markov process one can assign an \textit{infinitesimal generator} defined as follows
\begin{equation*}
	\mc{A}f(z) = \lim_{t \downarrow 0} \frac{E\left[f(Z_t) \mid Z_0 = z\right] - f(z)}{t}.
\end{equation*}
Intuitively, the quantity $\mc{A}f(z)$ is the mean infinitesimal rate of change in $f(Z_0)$, evolving according the process $Z_t$ starting at $Z_0=z$. Informally, if $z$ is chosen from the invariant measure, one expects the mean rate of change, $\int \mc{A}f(z) \mu(dz)$, to be zero and vice verse. This is formalized as the following proposition, which is the main result needed in the following sections.

\begin{prop} \label{prop:InvarianceThm} 
Let $\mc{A}$ be the infinitesimal generator associated to the process $Z_t$ and $\mu$ be a probability measure on $\mathcal{Z}$ such that
\begin{equation*}
	\int \mc{A}f(z) \mu(dz) = 0  \qquad \forall f\in \mc{C}, 
\end{equation*}
where $\mc{C}$ is a large enough class of functions. Then, $\mu$ is the invariant measure for the process $Z_t$.
\end{prop}

\subsection{Piecewise Deterministic Markov Processes}

The required background on the piecewise deterministic Markov processes (PDMP), which consists of a definition and a derivation of the infinitesimal generator, is developed in this section. As introduced in \cite{davis1984piecewise}, a piecewise deterministic Markov process is a stochastic process consisting of deterministic motion punctuated by Poisson jumps. Informally, a PDMP on a set $E$ is characterized by three objects, namely, a flow $\phi_t$, a jump rate $\lambda$, and a transition kernel $Q$. Starting at $z\in E$, it evolves according to the flow $\phi_s(z)$ until the first jump time $T_1$ occurs. $T_1$ corresponds to the first arrival time a inhomogeneous Poisson process with rate function $\lambda (\phi_s(z))$. More precisely, $T_1$ is has the following distribution
\begin{align*}
	\mb{P}\{T_1>t\} = \exp\left(-\int_0^t \lambda(\phi_s(z)) ds\right).
\end{align*}
The location of the process at time $T_1$ is drawn from the measure $Q(\phi_{T_1}(z),\cdot)$ and the process continues from this point according to the flow $\phi_s$ until the second jump time $T_2$, and so on. 

In what follows, we briefly define the PDMP and review its basic properties. For a detailed account of the regularity conditions needed see \cite{davis1984piecewise}. Let $\mc{I}$ be a countable set and $d:\; \mc{I} \rightarrow \mb{N}$ be a given function. For each $i\in \mc{I}$, let $M_i$ be an open set in the Euclidean space $\mb{R}^{d(i)}$. Then, the state space $E$ is defined as follows
\begin{equation*}
	E = \left\{z=(i,x) \mid i\in \mc{I}, x\in M_i \right\}.
\end{equation*}
The state of the process will be denoted $z_t = (i_t,x_t)$. The law of the process is determined by the following objects:
\begin{enumerate}
\item Flows $\left\{\phi_s^i(\cdot); i\in \mc{I}\right\}$, defined by the ordinary differential equation
  \begin{align*}
	  \frac{d}{dt} \phi_t^i(x) = F^i(\phi_t^i(x)) \text{ and } \phi_0^i(x) = x,
   \end{align*}
   for vector fields $\left\{F^i; i\in \mc{I}\right\}$.
\item A measurable rate function $\lambda: E \rightarrow \mb{R}_+$.
\item A transition kernel $Q(z;A)$ for $z\in E$ and $A\subset E$.
\end{enumerate}
Let $\partial M_i$ denote the boundary of $M_i$. For $z=(i,x)\in E$, define $t^\star(z)$ as the first time the flow hits the boundary $\partial M_i$, starting from $z$. That is
\begin{align*}
	t^\star (z) = \inf \left\{t>0: \phi^i_t(x)\in \partial M_i\right\}.
\end{align*}
The process $Z_t$ starting from $z$ can now be constructed as follows. Define the distribution function $F$ by
\begin{equation*}
	F(t) = \begin{cases}
		1- \exp\left(-\int_0^t \lambda(i,\phi^i_s(x)) ds \right) & \textnormal{ for } t< t^\star(z)\\
		1 & \textnormal{ for }  t\ge t^\star(z).
\end{cases}
\end{equation*}
Select a random variable $T_1 \sim F$. Independent of $T_1$, select $(i^\prime, x^\prime) \in E$ according to $Q\left(\left(i,\phi^i_{T_1}(x)\right); \cdot \right)$. The trajectory of $Z_t$ is given by
\begin{equation*}
	Z_t = \left(i_t, x_t\right) =	\begin{cases}
		\left(i,\phi^i_t(x)\right) & \textnormal{ for }t< T_1,\\
		\left(i^\prime, x^\prime\right)  & \textnormal{ for } t=T_1. \end{cases}
\end{equation*}
Starting from $Z_{T_1}$, sample the next inter-jump time $T_2-T_1$ and post-jump location $Z_{T_2}$ in a similar way. And so on. This defines a Markov process with deterministic paths between the jump times. We assume that $\mb{E}N_t <\infty$, where $N_t$ is the number of jumps in $[0,t]$. This holds under a mild regularity condition on $\lambda$, which is satisfied in all the cases considered in the paper, assuming the energy function (negative log-likelihood) is piecewise continuously differentiable.

The infinitesimal generator for PDMPs is given explicitly by \citet[Theorem 5.5]{davis1984piecewise} as follows.

\begin{prop}[Theorem 5.5 in \cite{davis1984piecewise}] \label{prop:generator}
Using the notation above, the infinitesimal generator of the process $z_t = (i_t,x_t)$ defined above is
\begin{equation*}
	\mc{A}f(i,x) = \big\langle\nabla_x f(i,x), F^i(x)\big\rangle 
	+ \lambda(i,x) \int_E \big[ f(j , y)- f(i,x)\big] Q\left((i,x); d(j,y)\right).
\end{equation*} 
\end{prop}


\section{A Family of Bouncy Hybrid MCMC Samplers} \label{sec:hmc}

This section formalizes the proposed family of samplers. Denote a target $d$-dimensional density as $\pi(x)$, $x\in\mathbb{R}^d$. We assume $\pi$ is continuously differentiable on its domain. Denote as $U(x)=-\log\pi(x)$, the negative log-density of the target distribution.

Building on the background and notation from Section \ref{sec:background}, consider the piecewise deterministic Markov process on $\mb{R}^{2d}$ ($|\mc{I}| = 1$) defined as follows.
\begin{enumerate}
\item The \textit{flow} is defined by the following system of differential equations:
\begin{equation} \label{eq:flow:equations}
\begin{aligned}
	\begin{split}
	\dot{x} &= v,\\
	\dot{v} &= -\nabla U(x) + g(x),
	\end{split}
\end{aligned}
\end{equation}
where $g: \mb{R}^d \rightarrow \mb{R}^d$ is a general vector field on $\mb{R}^d$.
\item The \textit{rate function} is given as
\begin{align*}
	\lambda(x,v) = \max\left\{0, \langle v,g(x)\rangle\right\} + \lambda_0,
\end{align*}
where $\lambda_0$ is a constant called the \textit{refreshment rate}.

\item With probability $\frac{\lambda_0}{\lambda(x,v)}$ we refresh the velocity, i.e.~draw the new velocity from a standard normal distribution in $d$ dimensions. With probability $1-\frac{\lambda_0}{\lambda(x,v)}$, the velocity gets updated using the kernel $Q$. 
The \textit{jump/transition kernel} $Q$ is a deterministic kernel which maps $(x,v)$ to $(x,R(x)v)$, with 
\begin{equation} \label{eq:kernel:R}
	R(x) v = \left( I_d - 2 \frac{g(x)g(x)^\top}{\|g(x)\|^2}\right) v = v - 2 \frac{\langle v, g(x)\rangle}{\|g(x)\|^2} g(x).
\end{equation}
\end{enumerate}

Assuming there are no refreshments ($\lambda_0=0$), the updates above can be interpreted as follows. Imagine a particle $x\in \mb{R}^d$ moving in an environment with the kinetic energy function $K(v) = v^\top v/2$ and the potential energy function $U(x)$, in the presence of an (vector field of) external force $g(x)$. The Hamiltonian corresponding to the kinetic and potential energy functions $K$ and $U$ is $H(x,v)=U(x)+K(v)$. The infinitesimal rate of change in the Hamiltonian, is then given by $\langle v,g(x)\rangle$. Then, $\max\{0,\langle v,g(x)\rangle\}$ is the rate of the work done (power) by the external force to ``climb'' the ``energy hill.'' The jump time can be interpreted as the time at which the total energy spent on climbing the energy hill reaches a priori sampled (exponentially distributed) ``energy budget.'' The jump kernel corresponds to an elastic collision against the infinitely heavy imaginary ``wall'' perpendicular to $g(x)$. Introducing refreshments into this interpretation is straightforward.

\begin{remark}
In an independent earlier work, \cite[Section 2.4]{vanetti2017piecewise} also proposed the algorithm above. At the time when the first version of our work came out we were unaware of their work. 	
\end{remark}

A simple application of Proposition \ref{prop:generator} gives the infinitesimal generator of the described process as 
\begin{equation} \label{eq:inf:gen}
\begin{aligned}
	\mc{A}f =& \left\langle\nabla_x f , v\right\rangle + \left\langle\nabla_v f , -\nabla U(x) + g(x)\right\rangle - \lambda(x,v) f(x,v) \\
	& +  \max\left\{0, \langle v,g(x)\rangle\right\} f(x,R(x)v) + \lambda_0 \int\limits_{v^\prime\in\mathbb{R}^d} f(x,v^\prime) \psi_d(d v^\prime),
\end{aligned}
\end{equation}
where $\psi_d$ is the density of $d$-dimensional standard normal distribution. The following proposition shows the target distribution $\pi(x)$ is stationary for the defined piecewise deterministic process.

\begin{prop}\textnormal{\textbf{[Invariant density of BHS]}}\label{prop:invariance}
	Under mild regularity conditions, for instance continuous differentiability of the energy function $U$ and integrability of $g$, the measure defined with the density $\rho(x,v) = \pi(x) \psi_d(v)$ is a stationary measure for the process $\left\{(X_t,V_t):t\geq 0\right\}$.
\end{prop}

\begin{proof}
To prove that $\rho(x,v) = \pi(x) \psi_d(v)$ is the invariant density, by Proposition \ref{prop:InvarianceThm}, it suffices to verify that
\begin{equation} \label{eqn:ForwardEqnGeneral}
	\int\limits_{(x,v)\in\mathbb{R}^{2d}} \mc{A}f(x,v) d\rho(x,v) = 0,
\end{equation}
for all $f \in \mc{D}(\mc{A})$. Using \eqref{eq:inf:gen}, this translates to
\begin{equation} \label{eqn:InvarianceGeneralSampler}
\begin{aligned}
	&\int\limits_{(x,v)\in\mathbb{R}^{2d}} \bigg(\left\langle \nabla_x f,v\right\rangle+\left\langle\nabla_v f , -\nabla U(x) + g(x)\right\rangle+\max\left\{0, \langle v,g(x)\rangle\right\} f(x,R(x)v) \\
 	& \qquad +\lambda_0 \int\limits_{v^{\prime}\in\mathbb{R}^d} f(x,v^\prime) \psi_d(d v^\prime)-\lambda(x,v) f(x,v) \bigg) d\rho(x,v) = 0.
\end{aligned}
\end{equation}
Integration by parts yields 
\begin{equation*}
	\int\limits_{(x,v)\in\mathbb{R}^{2d}}\left\langle\nabla_x f , v\right\rangle d \rho(x,v) = \int\limits_{(x,v)\in\mathbb{R}^{2d}}\left\langle\nabla U(x), v\right\rangle f(x,v) d \rho(x,v) \qquad
	\text{ and } 
\end{equation*}
\begin{equation*}	
\int\limits_{(x,v)\in\mathbb{R}^{2d}}\left\langle\nabla_v f,-\nabla U(x)+g(x)\right\rangle d \rho(x,v) = \int\limits_{(x,v)\in\mathbb{R}^{2d}}\left\langle  v,-\nabla U(x)+g(x)\right\rangle f(x,v)d\rho(x,v).
\end{equation*}
Therefore, the first two terms in \eqref{eqn:InvarianceGeneralSampler} simplify as
\begin{equation*}
\begin{aligned}
	&\int\limits_{(x,v)\in\mathbb{R}^{2d}}\bigl( \left\langle \nabla_x f , v\right\rangle+\left\langle\nabla_v f , -\nabla U(x)+g(x)\right\rangle\bigr)d\rho(x,v) \\
	&= \int\limits_{(x,v)\in\mathbb{R}^{2d}}\left\langle  v,g(x)\right\rangle f(x,v) d\rho(x,v).
\end{aligned}
\end{equation*}
Terms involving $\lambda_0$ cancel out trivially. The remaining term is
\begin{align*}
	\int\limits_{(x,v)\in\mathbb{R}^{2d}}\max\left\{0,\langle v,g(x)\rangle\right\} \bigl(f(x,R(x)v)-f(x,v)\bigr)  d \rho(x,v),
\end{align*}
which can be written as
\begin{align*}
	&\int\max\{0, \langle v,g(x)\rangle\} f(x,R(x)v) d\rho(x,v) -\int\max\{0,\langle v,g(x)\rangle\}f(x,v) d\rho(x,v)\\
	=&\int\max\{0,\langle R(x)^\top u, g(x)\rangle\} f(x,u)  d \rho(x,u)-\int\max\{0, \langle v,g(x)\rangle\}f(x,v) d \rho(x,v)
\end{align*}
(the integrals above are either over $(x,v)\in\mathbb{R}^d\times\mathbb{R}^d$ or over $(x,u)\in\mathbb{R}^d\times\mathbb{R}^d$),
where we used the change of variables $u = R(x)v$ and the fact that $\rho(x,v)$ remains invariant as $R(x)$ is a rotation matrix.
Since $\left\langle R(x)^\top u, g(x)\right\rangle = -\left\langle u, g(x)\right\rangle$, which holds because $R(x)$ is the reflection against the hyperplane perpendicular to $g(x)$, the above further equals
\begin{equation*}
\begin{aligned}
	&\int\limits_{(x,v)\in\mathbb{R}^{2d}} \bigl(\max\left\{0, -\langle v,g(x)\rangle\right\}-\max\left\{0,\langle v,g(x)\rangle\right\}\bigr) f(x,v) d \rho(x,v)\\
	=& -\int\limits_{(x,v)\in\mathbb{R}^{2d}}\left\langle v,g(x)\right\rangle f(x,v) d\rho(x,v).
\end{aligned}
\end{equation*}
 These together prove \eqref{eqn:ForwardEqnGeneral}.
\end{proof}

\begin{remarks} 
$\;$
\begin{itemize}[leftmargin=*]
\item Note that the proposition above still holds if we make any of the following changes to the sampler.
	\begin{itemize}
	\item Given a function $\gamma:\mathbb{R}^d\times\mathbb{R}^d$ satisfying $\gamma\left(x,R(x)^\top v\right)=\gamma(x,v)$, we can modify the rate function so that it becomes $\lambda(x,v)+\gamma(x,v)$ with $\lambda(x,v)$ given above.
	\item $R(x)$ can be any rotation matrix satisfying $\left\langle R(x)^\top u, g(x)\right\rangle=-\left\langle u, g(x)\right\rangle$. 
	\end{itemize}
	Further modifications and generalizations are given in Section \ref{sec:extensions}.
	
\item The ideas in \cite{afshar2015reflection} can be applied to here as well. This will particularly introduce a Bouncy Hybrid Sampler for energy functions with discontinuities.	
\end{itemize}
\end{remarks}

\subsection{Special Instances of $g$}

We explain how some of the existing samplers are the special instances of the given class of samplers.

\begin{itemize}
\item \textbf{Randomized Hamiltonian Monte Carlo.} For $g(x)=0$, the piecewise deterministic algorithm described becomes the  Randomized Hamiltonian Monte Carlo (RHMC) of \cite{bou2017randomized}. In this case, the flow equations in \eqref{eq:flow:equations} exactly correspond to the Hamiltonian dynamics.
The rate function $\lambda(x,v)$ becomes a constant, hence there is no bouncing of the particle but only refreshment. At every step of the algorithm, the moving time $\tau$ becomes a random variable drawn from $\textnormal{Exp}(\lambda_0)$. Thus at every step, the duration of the Hamiltonian flow is an independent exponential random variable unlike in the classical Hamiltonian MCMC where the duration of the Hamiltonian flow is fixed in advance \citep{duane1987hybrid, neal2011mcmc}. The velocity at every step gets drawn from the standard normal distribution in $d$ dimensions, corresponding to the complete momentum randomization in the RHMC. By slightly modifying the velocity update after refreshment events, we get that the proposed Bouncy Hybrid Sampler completely generalizes RHMC (Section \ref{sec:RHMC:velocity}).
\vspace{1em}

\item \textbf{Bouncy Particle Sampler.} 
For $g(x)=\nabla U(x)$, the proposed piecewise deterministic MCMC becomes the Bouncy Particle Sampler of \cite{bouchard2015bouncy}. In this case, the particle moves along straight lines with constant velocity $v$ in between bouncing or refreshments events. The duration time along each piecewise linear path is a modeled as the first arrival time of a Poisson process with rate $\lambda(x,v)$, where $x$ and $v$ are the current position and time of the moving particle. 

Note that the refreshment events ($\lambda_0>0$) are needed for the BPS to be ergodic. It is possible to remove the refreshments and still keep the ergodicity of the chain by introducing the stochastic velocity update \citep{wu2017generalized}. We elaborate more on this modification in Section \ref{sec:general:BPS}.

\end{itemize}

\subsection{Exactly Solvable Flows}

When simulating the flow numerically, we need a symmetric flow because, in general, computing the Metropolis-Hastings filter is impossible, as the transition kernel is hard to compute (practically impossible). In general, the leapfrog integrator can be used, it is time-reversible and volume preserving. The problem with Metropolis correction is the non-reversibility of the jump process.
This advocates for the choices with exactly solvable flows, like the Quadratic Bouncy Hybrid Sampler introduced in Section \ref{sec:tmvg}. 

When there are explicit formulas available for the deterministic flows, sampling of the jump process is studied by \citet{lemaire2016exact} and a method is suggested for exact sampling of the jump process. Their algorithm is based on the famous thinning procedure for simulation of inhomogeneous Poisson processes, developed in \cite{lewis1979simulation}. We elaborate on computational aspects further in the rest of the paper.


\section{BHS Algorithm} \label{sec:algorithm}

Before presenting the specific examples of the sampler, we describe the algorithm in detail. Given the gradient $\nabla U$ of the negative logarithm of the target density in $d$ dimensions, the parameters of the algorithm are: a refreshment constant $\lambda_0$, function $g:\mathbb{R}^d\rightarrow\mathbb{R}^d$, the total time of the process $T_{total}$ and the discretizing time interval length $\delta$.
Denote the current time as $t_{curr}$, where initially $t_{curr}\leftarrow 0$ and the current position and velocity as $(X_{t_{curr}}, V_{t_{curr}})$. While $t_{curr}<T_{total}$, the algorithm consists of repeating the following steps.

\begin{enumerate}[leftmargin=*]
	\item \label{item:app:flow} Solve the flow equations in \eqref{eq:flow:equations} s.t.~the initial solution at $t=0$ is $\left(X_{t_{curr}},V_{t_{curr}}\right)$. This gives a unique set of functions $x_t$ and $v_t$.
	
	\item Set the rate of a Poisson process as
	\begin{equation*}
		\bar{\lambda}(x_t,v_t)=\max\left\{0, g(x_t)\cdot v_t\right\},	
	\end{equation*}
	with $(x_t,v_t)$ being the solutions from the step \ref{item:app:flow} above.
	
	\item To compute the moving time $\tau$ do the following.
	\begin{enumerate}[leftmargin=*]
	\item\label{item:app:boncing:time} Sample bounce time $\tau^B$ as the first arrival time of the Poisson process with the rate $\bar{\lambda}(x_t,v_t)$ above, i.e.~sample $\tau^B$ from the distribution satisfying
	\begin{equation*}
		\mathbb{P}\left\{\tau^B\geq t\right\}=\exp\left(-\int_0^t \left(g(x_t)^\top v_t\right)_+\right), \qquad t\geq 0.
	\end{equation*}
 
	\item Sample a refreshment time $\tau^R\sim\textnormal{Exp}(\lambda_0)$. 
	\item Take the moving time to be the minimum of the two $\tau=\min\left\{\tau^B,\tau^R\right\}$.
	\end{enumerate}
	
	\item The particle moves along $x_t$ with velocity $v_t$ for the total time of $\tau$ so that $(X_t,V_t)=\left(x_{t-t_{curr}}, v_{t-t_{curr}}\right)$, $t\in\left[t_{curr}, t_{curr}+\tau\right]$.
	 After time $\tau$, the current time gets updated according to $t_{curr}\leftarrow t_{curr}+\tau$ and the current state of the process becomes $\left(X_{t_{curr}}, V_{t_{curr}}\right)=(x_{\tau}, v_{\tau})$.
	 
	\item After moving along the deterministic flow for a random time $\tau$, the velocity gets updated depending on which of the following two events happened. 
	\begin{enumerate}
	\item If $\tau=\tau^B$ (bouncing event), the current velocity gets updated using the jump kernel $Q$ evaluated at the current point $(X_{t_{curr}}, V_{t_{curr}})$. 
	 \item If $\tau=\tau^R$ (refreshment), the new velocity gets drawn from the standard normal distribution in $d$ dimensions. 
	 \end{enumerate}
\end{enumerate}

We elaborate on Step \eqref{item:app:boncing:time} of the algorithm above. There are several possible ways to sample the bouncing time $\tau^B$. We decide which one to use depending on the computational complexity of the sampling governed by the target density and the choice of $g$.
\begin{itemize}
\item \textbf{Inverse transform sampling.} $\tau^B$ can be computed as the solution of
\begin{equation*}
	-\log V=\int_0^{\tau^B}\left(g(x_t)^\top v_t\right)_+dt,
\end{equation*}
	where $V\sim\textnormal{Unif}(0,1)$ independent of the process so far. We use the inverse transform sampling when the target density is univariate Gaussian and $g$ is linear (Section \ref{sec:univariate:gaussian}).
	Since the above equation might be hard to solve fast in general, we might use some of the techniques presented next for other choices of $\nabla U(x)$ and $g(x))$.
 
\item \textbf{Thinning method of \cite{lewis1979simulation}.} Assume we have an upper bound $\Lambda(t)$ on the rate function, $\Lambda(t)\geq \bar{\lambda}(x_t,v_t)$, $\forall t\geq 0$, called thinning proposal, for which sampling the arrival times $\tau_1,\tau_2,\ldots$ is easy. Delete the time $\tau_k$ with probability $1-\frac{\bar{\lambda}(x_{\tau_k},v_{\tau_k})}{\Lambda(\tau_k)}$ for each $k\geq 1$. The smallest $k$ for which $\tau_k$ was not deleted becomes the bouncing time $\tau^B$. We illustrate this method when sampling from the truncated multivariate normal distribution in Section \ref{sec:thinning}, where we derive a constant bound on the rate function.

\item \textbf{Approximate thinning proposal of \cite{pakman2016stochastic}.} Since a possible upper bound on the rate function $\bar{\lambda}(x_t,v_t)$ might be conservative, the thinning method might be slow. \cite{pakman2016stochastic} propose an adaptive approximate thinning proposal rate, which introduces small bias but provides computational gain. 
\end{itemize}

Since our proposed Bouncy Hybrid Sampler is a continuous-time Markov chain, it is worth explaining how to get the discrete samples from the simulated trajectory of the chain. Denote the positions of the chain as $X_t$, $0\leq t\leq T_{total}$ and denote the skeleton of the simulated trajectory as $\left(X^{(i)}, V^{(i)}\right)$, $i=1,\ldots, n$, representing points at which the particle switches trajectories. In other words, these points present the current position and velocity of the particle after running steps 1-4 of the algorithm above. Since the piecewise deterministic process above is a continuous process, we cannot only use the skeleton of the above algorithm for estimation. 
More specifically, for estimating $\int_{x\in\mathbb{R}^d}\phi(x)\pi(x)dx$ for a given function $\phi$, the estimator $\frac{1}{n}\sum_{i=1}^n \phi(X^{(i)})dt$ is biased unlike the following $\frac{1}{T_{total}}\int_0^{T_{total}} \phi(X_t)dt.$
When the latter integral is not tractable, we divide $\left[0,T_{total}\right]$ in regular time intervals of fixed length $\delta>0$ to obtain the estimator
\begin{equation*}
	\frac{1}{N}\sum_{i=1}^N\phi(X_{i\delta}),	
\end{equation*}
where $N=\lfloor T_{total}/\delta\rfloor$. This introduces another parameter $\delta$ in our algorithm. Practically, the estimator above is pretty robust to the choices of $\delta$.

\subsection{Univariate Gaussian} \label{sec:univariate:gaussian}

To illustrate the above algorithm, we present the details for sampling from $\mathcal{N}(0,1)$ distribution and a linear function $g$.
In this example, $g(x)=ax$, $a\in\mathbb{R}$, and $U(x)=\frac{x^2}{2}$, thus the system of differential equations becomes 
\begin{equation*}
\begin{aligned}
	\dot{x}_t&=v_t \\
	\dot{v}_t&=-x_t+ax_t,
\end{aligned}	
\end{equation*}
implying $\ddot{x}_t-x_t(a-1)=0$. We have the following solution for $x_t$
\begin{equation*}
x_t = \begin{cases}
	C_1e^{t\sqrt{a-1}}+C_2e^{-t\sqrt{a-1}} 
	& \textnormal{ for } a>1 \\
	C_1\cos(t\sqrt{1-a})+C_2\sin(t\sqrt{1-a})  
	& \textnormal{ for } a<1 \\
	C_1t+C_2 
	& \textnormal{ for } a=1
\end{cases}	
\end{equation*}
and $v_t$ 
\begin{equation*}
v_t = \begin{cases}
	C_1\sqrt{a-1}e^{t\sqrt{a-1}}-C_2\sqrt{a-1}e^{-t\sqrt{a-1}} 
	& \textnormal{ for } a>1 \\
	-C_1\sqrt{1-a}\sin(t\sqrt{1-a})+C_2\sqrt{1-a}\cos(t\sqrt{1-a})  
	& \textnormal{ for } a<1 \\
	C_1 
	& \textnormal{ for } a=1	.
\end{cases}
\end{equation*}
Taking into account the initial condition $x_0$ and $v_0$, we have
\begin{equation*}
	(C_1,C_2)=\begin{cases}
		\left(x_0+\frac{v_0}{\sqrt{a-1}}, x_0-\frac{v_0}{\sqrt{a-1}}\right) & \textnormal{ for } a>1 \\
		\left(x_0, \frac{v_0}{\sqrt{1-a}}\right) & \textnormal{ for }a<1 \\
		(v_0, x_0)& \textnormal{ for } a=1.		
	\end{cases}	
\end{equation*}

The bouncing time is $\tau^B$ for which
\begin{equation*}
	-\log V = \int_0^{\tau^B}(v_t\cdot g(x_t))_+dt=\int_0^{\tau^B}\left(ax_t\cdot\frac{dx_t}{dt}\right)_+dt,
\end{equation*}
where $V\sim\textnormal{Unif}(0,1)$.
We compute $\tau^B$ separately for each $a>1$, $a<1$ and $a=1$.
\begin{enumerate}[leftmargin=*]
	\item $a>1$. $\tau^B$ solves
	\begin{equation*}
	-\log V=\int_0^{\tau^B}a\sqrt{a-1}\left(C_1^2e^{2t\sqrt{a-1}}-C_2^2e^{-2t\sqrt{a-1}}\right)_+dt.
	\end{equation*}
	Function $a\sqrt{a-1}\left(C_1^2e^{2t\sqrt{a-1}}-C_2^2e^{-2t\sqrt{a-1}}\right)$ above is an increasing function (first derivative positive) and achieves zero at $t_0=\frac{\ln (C_2^2/C_1^2)}{4\sqrt{a-1}}$.
	Then $\tau^B\geq \max\{t_0,0\}={t_0}_+$ solves the following
	\begin{equation*}
		-\log V+\frac{a}{2}\left(C_1^2e^{2{t_0}_+\sqrt{a-1}}+C_2^2e^{-2{t_0}_+\sqrt{a-1}}\right)=\frac{a}{2}\left(C_1^2e^{2\tau^B\sqrt{a-1}}+C_2^2e^{-2\tau^B \sqrt{a-1}}\right).
	\end{equation*}

	\item $a<1$. The path and velocity functions $x_t$ and $v_t$ can be written as
	\begin{equation*}
	\begin{aligned}
		x_t&=r\cos\left(c+t\sqrt{1-a}\right) \\
		v_t&=-r\sqrt{1-a}\sin\left(c+t\sqrt{1-a}\right),
	\end{aligned}
	\end{equation*}
	where $r=\sqrt{C_1^2+C_2^2}$ and $c$ satisfies $\cos c=\frac{C_1}{r}$ and $\sin c=\frac{-C_2}{r}$.
	The bouncy time $\tau^B$ then solves
	\begin{equation*}
		-\log V=\int_0^{\tau^B}\left(\frac{-ar^2\sqrt{1-a}}{2}\sin(2c+2t\sqrt{1-a}))\right)_+dt.
	\end{equation*}
	The period of the sine function above is $\pi/\sqrt{1-a}$.
	Since
	\begin{equation*}
	\begin{aligned}
		&\int_0^{\frac{\pi}{\sqrt{1-a}}} \left(\sin(2c+2t\sqrt{1-a})\right)_+dt 
		=\frac{1}{\sqrt{1-a}},
	\end{aligned}
	\end{equation*}
	we have that the integral of the rate function across one period is
	\begin{equation*}
	\int_0^{\frac{\pi}{\sqrt{1-a}}}	\left(\frac{-ar^2\sqrt{1-a}}{2}\sin(2c+2t\sqrt{1-a}))\right)_+dt = \frac{|a|r^2}{2}.
	\end{equation*}
	Thus, $\tau^B$ is the solution of
	\begin{equation} \label{eq:gaussian:leftover:integral}
	-\log V-n\frac{|a|r^2}{2}=\int_{n\frac{\pi}{\sqrt{1-a}}}^{\tau^B}\left(\frac{-ar^2\sqrt{1-a}}{2}\sin(2c+2t\sqrt{1-a}))\right)_+dt
	\end{equation}
	for $n=\lfloor\frac{-\log V}{|a|r^2/2}\rfloor$. Denote the RHS of the above equation as $L$ and the integrand on the LHS above without the positive part as $h(t)=\frac{-ar^2\sqrt{1-a}}{2}\sin(2c+2t\sqrt{1-a}))$.
	To solve the equation \eqref{eq:gaussian:leftover:integral}  with respect to $\tau^B$, we find the two zeros $t_1$ and $t_2$ of the function $\sin(2c+2t\sqrt{1-a})$ that are in between $\frac{n\pi}{\sqrt{1-a}}$ and $\frac{(n+1)\pi}{\sqrt{1-a}}$. The zeros of $\sin(2c+2t\sqrt{1-a})$ are of form $t_1=\frac{\pi k_1-2c}{2\sqrt{1-a}}$ and $t_2 = \frac{\pi k_2-2c}{2\sqrt{1-a}}$, $k_1, k_2\in\mathbb{Z}$. $k_1=\lceil{2n+\frac{2c}{\pi}}\rceil$ and $k_2=k_1+1$. We compute the integral 
	\begin{equation*}
		I_p=\int_{\frac{n\pi}{\sqrt{1-a}}}^{t_1}\left(\frac{-ar^2\sqrt{1-a}}{2}\sin(2c+2t\sqrt{1-a}))\right)dt.	
	\end{equation*}
	Depending on the value of the integral above we differentiate between three cases to compute $\tau^B$.
	\begin{enumerate}[leftmargin=*]
	\item $I\leq 0$. In this case $h(t)$ is positive for $t\in[t_1, t_2]$, thus $\tau^B\in[t_1, t_2]$ solves
		$\int_{t_1}^{\tau^B}h(t)dt=L.$
	
	\item $I\leq L$. In this case $h(t)$ is positive for $t\in\left[\frac{n\pi}{\sqrt{1-a}}, t_1\right]$ and $\tau^B\in\left[\frac{n\pi}{\sqrt{1-a}}, t_1\right]$ solves
	$\int_{\frac{n\pi}{\sqrt{1-a}}}^{\tau^B}h(t)dt=L.$

	\item $I\geq L$. In this case, the solution $\tau^B\in\left[t_2,\frac{(n+1)\pi}{\sqrt{1-a}}\right]$ satisfies $\int_{t_2}^{\tau^B}h(t)=L-I.$	

	\end{enumerate}

	\item $a=1$. $\tau^B$ solves
	\begin{equation*}
		-\log V=\int_0^{\tau^B}\left(C_1^2t+C_1C_2\right)_+dt	
	\end{equation*}
	Function $C_1^2t+C_1C_2$ is increasing and achieves zero at $t_0=-\frac{C_2}{C_1}$. Then $\tau^B\geq\max\{t_0,0\}={t_0}_+$ solves the quadratic equation
	\begin{equation*}
		-\log V+\frac{C_1^2}{2}{t_0}_+^2+C_1C_2{t_0}_+= \frac{C_1^2}{2}(\tau^B)^2+C_1C_2\tau^B.
	\end{equation*}	
\end{enumerate}

To illustrate the convergence of the samples we get by running the above sampler to the standard normal distribution, we plot the histograms of the samples for two different functions $g$, $g(x)=-x$ and $g(x)=x$ (Figure \ref{fig:univariate:gaussian}).

\begin{center}
\begin{figure}[h!]
\begin{subfigure}{.5\textwidth}
  \centering
  \includegraphics[width=\linewidth]{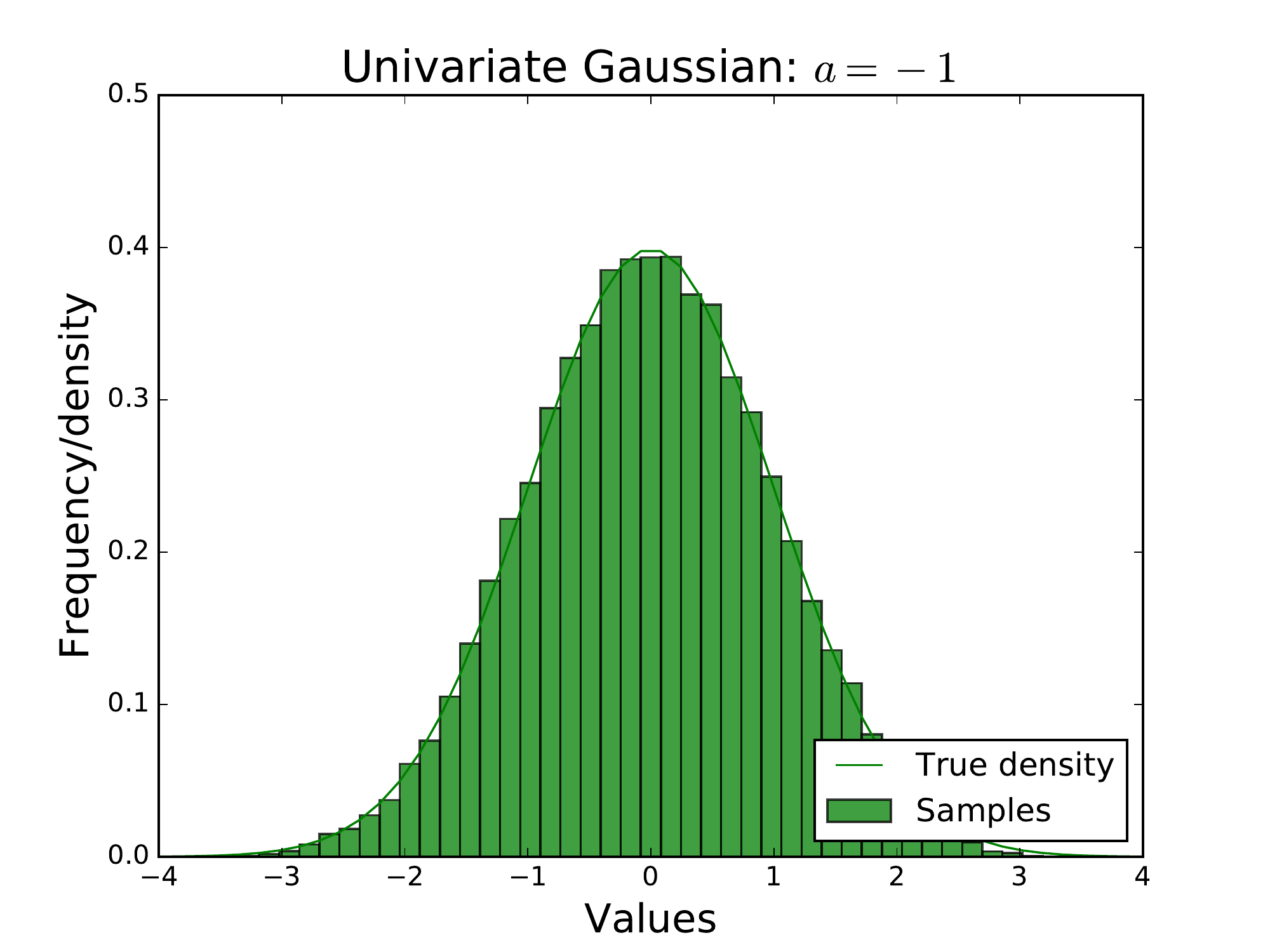}
  \caption{$a=-1$}
  \label{fig:sub1}
\end{subfigure}%
\begin{subfigure}{.5\textwidth}
  \centering
  \includegraphics[width=\linewidth]{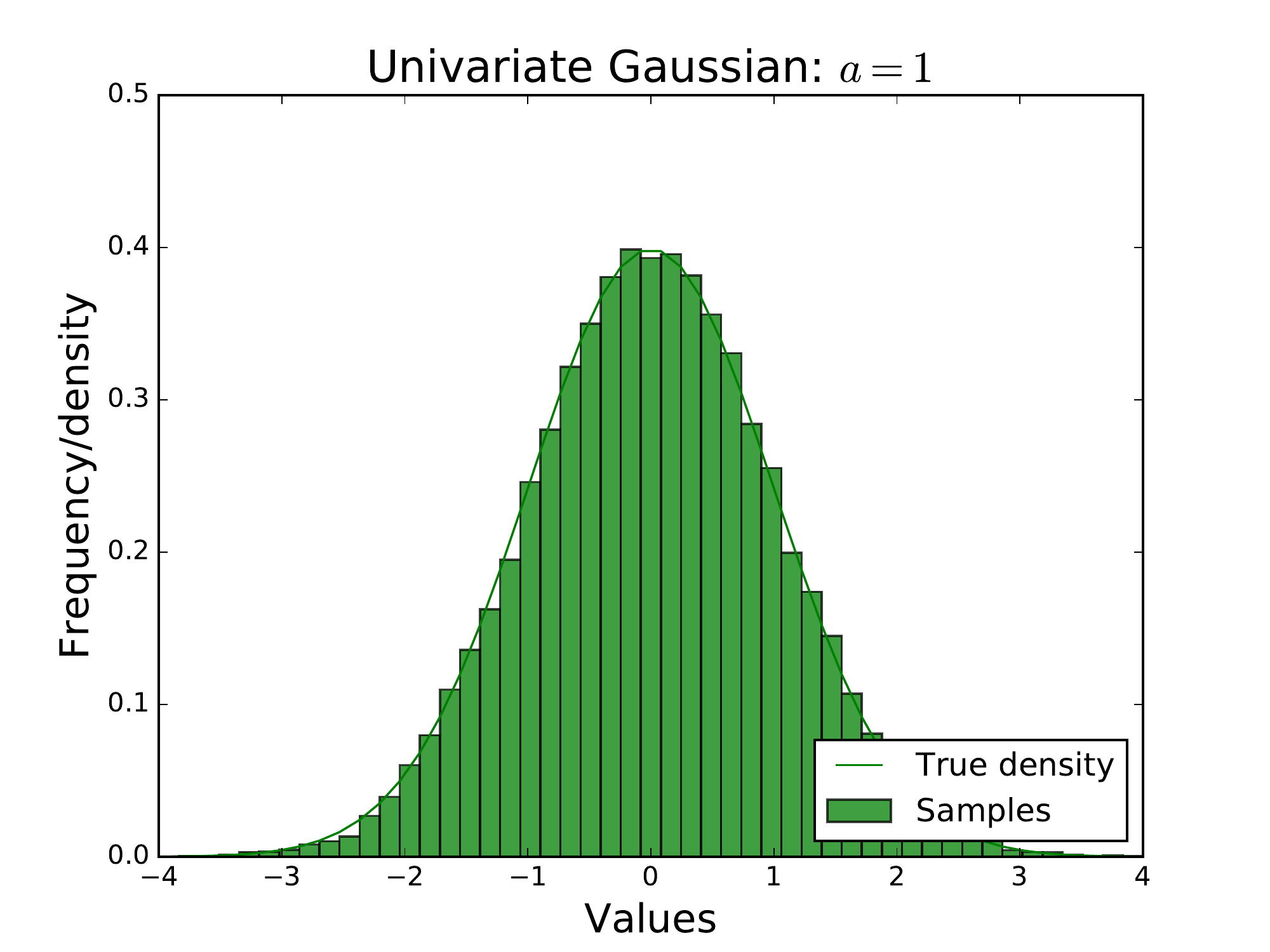}
  \caption{$a=1$}
  \label{fig:sub2}
\end{subfigure}
\caption{Histograms of samples by running BHS chain for two different linear functions $g(x)=-x$ (left) and $g(x)=x$ (right).}
\label{fig:univariate:gaussian}
\end{figure}
\end{center}


\section{Quadratic Bouncy Hybrid Sampler} \label{sec:tmvg}

We elaborate on a special instance of the proposed Bouncy Hybrid Sampler specialized for sampling $X\sim\mathcal{N}(\mu,\Sigma)$, $\mu\in\mathbb{R}^d$ and $\Sigma\in\mathbb{R}^{d\times d}$, in $d$ dimensions with affine constrains 
\begin{equation*}
	F^\top X+h\geq 0,
\end{equation*}
where $F\in\mathbb{R}^{d\times m}$ and $h\in\mathbb{R}^m$ for some number of constraints $m$. We call the derived MCMC the Quadratic Bouncy Hybrid Sampler (QBHS).

Previously, this problem was considered by \cite{pakman2014exact} who provide the exact solution for the Hamiltonian dynamics corresponding to a quadratic energy function. Similarly to their sampler, our flow equations have the exact solution. Unlike their sampler, the duration of each flow in our case is random.

\subsection{Sampling from Unrestricted Multivariate Normal Distribution} 

We start by explaining how to run the BHS for unrestricted $\mathcal{N}(\mu,\Sigma)$ distribution. 
In this case the gradient of the negative log-density is $\nabla U(x)=\Sigma^{-1}(x-\mu)$ so the system becomes
\begin{equation*}
\begin{aligned}
	\dot{x}&=v\\
	\dot{v}&=-\Sigma^{-1}(x-\mu)+g(x)
\end{aligned}	
\end{equation*}
with the initial solution $x_0$ and $v_0$.
This implies $\ddot{x}=-\Sigma^{-1}x+g(x)+\Sigma^{-1}\mu$. Take 
\begin{equation*}	
	g(x)=Ax
\end{equation*}
for a matrix $A\in\mathbb{R}^{d\times d}$, where $\Sigma^{-1}-A$ is a diagonalizable matrix, hence can be written as $\Sigma^{-1}-A=-P^{-1}A_dP$ with $A_d$ diagonal matrix with non-zero elements on its diagonal and $P$ invertible. Denote the non-zero diagonal elements of $A_d$ as $a_1, \ldots, a_d$.
Then the differential equation becomes $\ddot{x}-P^{-1}A_dPx-\Sigma^{-1}\mu=0$, or equivalently $P\ddot{x}-A_dPx-P\Sigma^{-1}\mu=0$. Changing the variables $y=Px$, the equation becomes $\ddot{y}-A_dy-P\Sigma^{-1}\mu=0$. The solution $y=y_t=(y_{1,t}, \ldots, y_{d,t})$ is given by
\begin{equation*}
	y_{k,t}
	=\begin{cases} C_{k,1}\cos(\sqrt{-a_k}t)+C_{k,2}\sin(\sqrt{-a_k}t)-\frac{1}{a_k}(P\Sigma^{-1}\mu)_k & \textnormal{ for } a_k<0 \\
	C_{k,1}e^{\sqrt{a_k}t}+C_{k,2}e^{-\sqrt{a_k}t}-\frac{1}{a_k}(P\Sigma^{-1}\mu)_k & \textnormal{ for }a_k>0, 
\end{cases}
\end{equation*}
and $\dot{y}$ is given by
\begin{equation*}
	\dot{y}_{k,t}=\begin{cases} -C_{k,1}\sqrt{-a_k}\sin(\sqrt{-a_k}t)+C_{k,2}\sqrt{-a_k}\cos(\sqrt{-a_k}t) & \textnormal{ for } a_k<0 \\
	C_{k,1}\sqrt{a_k}e^{\sqrt{a_k}t}-C_{k,2}\sqrt{a_k}e^{-\sqrt{a_k}t} &\textnormal{ for }a_k>0,
\end{cases}
\end{equation*}
for $k=1,\ldots, d$, where $(C_{k,1}, C_{k,2})$ are chosen to satisfy the initial condition $y_0=Px_0$ and $\dot{y}_0=Pv_0$. Hence,
\begin{equation} \label{eq:tmvg:constants:C}
	(C_{k,1}, C_{k,2})=\begin{cases} \left((Px_0)_k+\frac{1}{a_k}(P\Sigma^{-1}\mu)_k, \frac{1}{\sqrt{-a_k}}(Pv_0)_k\right)
	& \textnormal{ for }a_k<0 \\
	\frac{1}{2}\left((Px_0)_k+\frac{1}{a_k}(P\Sigma^{-1}\mu)_k\pm\frac{1}{\sqrt{a_k}}(Pv_0)_k \right)
	& \textnormal{ for }a_k>0,
	\end{cases}	
\end{equation}
where in the case of $a_k>0$, $C_{k,1}$ takes the plus sign and $C_{k,2}$ takes the minus sign above.

Given the solution for $y$ and $\dot{y}$, the solution for the position and velocity becomes $(x,v)=\left(P^{-1}y, P^{-1}\dot{y}\right)$. The rate function $\bar{\lambda}(x_t,v_t)=\left(g(x_t)^\top v_t\right)_+$ is
\begin{equation*} 
	\bar{\lambda}(x_t,v_t)=\left(y_t^\top {P^{-1}}^\top A^\top P^{-1}\dot{y}_t\right)_+.
\end{equation*}
Since sampling $\tau^B$ via inverse transform is hard in this case, we compute $\tau^B$ using the thinning method as follows.

\subsection{Thinning Method} \label{sec:thinning}

To employ the thinning method to sample $\tau^B$ we need an upper bound on the rate function $\bar{\lambda}(x_t,v_t)$. We use the following upper bound
\begin{equation*}
	\bar{\lambda}(x_t,v_t)\leq \|y_t\|_2\left\|{P^{-1}}^\top A^\top P^{-1}\right\|_2\|\dot{y}_t\|_2,
\end{equation*}
further noting
\begin{equation*}
\begin{aligned}
	\left(C_{k,1}\cos(\sqrt{-a_k}t)+C_{k,2}\sin(\sqrt{-a_k}t)\right)^2\leq \max\left\{|C_{k,1}|, |C_{k,2}|\right\}^2+\left|C_{k,1}C_{k,2}\right| = B_k, \\
	\left(-C_{k,1}\sin(\sqrt{-a_k}t)+C_{k,2}\cos(\sqrt{-a_k}t)\right)^2 \leq \max\left\{|C_{k,1}|, |C_{k,2}|\right\}^2 +\left|C_{k,1}C_{k,2}\right|.
\end{aligned}
\end{equation*}
Assuming $a_k\leq 0$ for all $k=1,\ldots, d$, we get
\begin{equation*}
\begin{aligned}
	y_{k,t}^2 &\leq B_k+2 \sqrt{B_k}\left|\frac{(P\Sigma^{-1}\mu)_k}{a_k}\right|+\frac{(P\Sigma^{-1}\mu)_k^2}{a_k^2} \\
	\dot{y}_{k,t}^2 &\leq (-a_k)B_k,
\end{aligned}
\end{equation*}
implying
\begin{equation*}
\begin{aligned}
	\|y_t\|_2^2 &=\sum_{k=1}^d y_{k,t}^2 \leq  \sum_{k=1}^d \left(B_k+2\sqrt{B_k}\left|\frac{(P\Sigma^{-1}\mu)_k}{a_k}\right|+\frac{(P\Sigma^{-1}\mu)_k^2}{a_k^2}\right)	 \\
	\|\dot{y}_t\|_2^2 &= \sum_{k=1}^d \dot{y}_{k,t}^2\leq \sum_{k=1}^d (-a_k)B_k.
\end{aligned}	
\end{equation*}
Thus, an upper bound on the rate becomes a constant (not depending on $t$)
\begin{equation*}
\begin{aligned}
	\Lambda(t)=\Lambda &= \sqrt{\sum_{k=1}^d \left(B_k+2\sqrt{B_k}\left|\frac{(P\Sigma^{-1}\mu)_k}{a_k}\right|+\frac{(P\Sigma^{-1}\mu)_k^2}{a_k^2}\right)	} \\
	& \qquad \cdot\left\|{P^{-1}}^\top A^\top P^{-1}\right\|_2
	 \cdot \sqrt{\sum_{k=1}^d (-a_k)B_k}.
\end{aligned}
\end{equation*}

\subsection{Sampling from a Truncated Normal Distribution} 

We now incorporate the constraints in the sampler. The constraints on $x$ implies the constraints on $y=Px$ are 
\begin{equation*}
	F^\top P^{-1}y+h\geq 0.
\end{equation*}
Denote the vectors $C_1=(C_{1,1},\ldots, C_{d,1})^\top\in\mathbb{R}^d$, $C_2=(C_{1,2},\ldots, C_{d,2})\in\mathbb{R}^d$, $o=A_d^{-1}(P\Sigma^{-1})\mu=\left(\frac{(P\Sigma^{-1}\mu)_1}{a_1},\ldots, \frac{(P\Sigma^{-1}\mu)_d}{a_d}\right)\in\mathbb{R}^d$, and the matrix $F^\top P^{-1}=K^\top$ with the columns of $K$ as $K_1,\ldots, K_m$. 
Assuming $a_1=\ldots=a_d=a<0$, each of the $m$ constraints can be written as
\begin{equation*}
\begin{aligned}
	&K_j^\top y+h_j = \sum_{i=1}^d K_{i,j}y_{i,t}+h_j \\
	&=\sum_{i=1}^d K_{i,j}\left(C_{i,1}\cos(\sqrt{-a_i}t)+C_{i,2}\sin(\sqrt{-a_i}t)-\frac{1}{a_i}(P\Sigma^{-1}\mu)_i\right)+h_j \\
	&=\left(\sum_{i=1}^dK_{i,j}C_{i,1}\right)\cos(\sqrt{-a}t)+\left(\sum_{i=1}^dK_{i,j}C_{i,2}\right)\sin(\sqrt{-a}t)
	-\sum_{i=1}^d K_{i,j} \frac{(P\Sigma^{-1}\mu)_i}{a_i}+h_j\\
	&=u_j\cos(\sqrt{-a}t+\phi_j)+q_j\geq 0 ,
\end{aligned}
\end{equation*}
where
\begin{equation*}
\begin{aligned}
	u_j&=\sqrt{\left(K_j^\top C_1\right)^2+\left(K_j^\top C_2\right)^2}, \\
	\cos\phi_j &= \frac{K_j^\top C_1}{u_j},\;\; \sin\phi_j=-\frac{K_j^\top C_2}{u_j}, \\
	q_j&=-K_j^\top o+h_j,
\end{aligned}	
\end{equation*}
for $j=1,\ldots, m$. The above is satisfied by taking $\phi_j=-\textnormal{sign}(-K_j^\top C_2)\cdot\textnormal{arccos}\left(\frac{K_j^\top C_1}{u_j}\right)$, $j=1,\ldots,m$.
 Denote the set of the reachable constraints as $\mathcal{R}=\left\{j\in\{1,\ldots, m\}:u_j>|h_j|\right\}$. Thus, the bouncing time $\tau^B$ is strictly smaller than 
\begin{equation} \label{eq:tmvg:time:bound}
	\tau^{BB}=\frac{1}{\sqrt{-a}}\min\left\{\textnormal{arccos}\left(-\frac{q_j}{u_j}\right)-\phi_j: j\in\mathcal{R}\right\}.
\end{equation}
Denote the index $j$ for which the minimum above is achieved as $j^*$, i.e.~$\tau^{BB}=\frac{1}{\sqrt{-a}}\left(\textnormal{arccos}(-q_{j^*}/u_{j^*})-\phi_j\right)$. This implies at time $\tau^{BB}$, the particle hits the wall $j^*$. When the particle hits the wall, its velocity reflects against the wall $F_{j^*}$, i.e.~its component perpendicular to the wall changes sign. Precisely, given the  velocity $v_t$ at the time $t$ of hitting the wall, the updated updated velocity becomes
\begin{equation} \label{eq:tmvg:reflected:velocity}
	v_t \leftarrow v_t-2\cdot \textnormal{proj}_{F_{j^*}}v_t = v_t - 2\frac{\langle v_t, F_{j^*}\rangle}{\|F_{j^*}\|_2^2} v_t.
\end{equation}

\subsection{QBHS Algorithm} 

Here is the summary of the algorithm for sampling from the truncated Gaussian distribution with mean $\mu\in\mathbb{R}^d$, variance $\Sigma\in\mathbb{R}^{d\times d}$ and the constraints described via a matrix $F\in\mathbb{R}^{d\times m}$ and a vector $h$.
Input parameters for the sampler are the following: 
\begin{enumerate}[label=(\roman*)]
\item matrix $P$ and $a_1=\ldots=a_d=a<0$, specifying function $g$;
\item refreshment rate $\lambda_0$;
\item total running time of the sampler $T_{total}$ and a constant time interval length $\delta$ based on which we collect discrete samples;
\item an initial point $(X_0,V_0)\in\mathbb{R}^d\times\mathbb{R}^d$ that satisfies the constraints.
\end{enumerate}
The current time is denoted as $t_{curr}$, where initially $t_{curr}\leftarrow 0$. At each step of the sampler while $t_{curr}<T_{total}$ repeat the following steps.

\begin{enumerate}
\item \label{item:tmvg:constants} Compute the vectors of constants $C_1\in\mathbb{R}^d$ and $C_2\in\mathbb{R}^d$ by solving \eqref{eq:tmvg:constants:C} corresponding to negative $a$ values with the initial position and velocity set at the current solution $\left(X_{t_{curr}}, V_{t_{curr}}\right)$.

\item Compute the following times:
\begin{enumerate}
\item bouncing time $\tau^B$ via thinning method;
\item reflection bound $\tau^{BB}$ from \eqref{eq:tmvg:time:bound};
\item refreshment time $\tau^{R}\sim\textnormal{Exp}(\lambda_0)$.
 \end{enumerate}
 
 \item Set the moving time $\tau$ to be the smallest of the above times, i.e. 
 \begin{equation*}
 	\tau=\min\left\{\tau^{B},\tau^{BB}, \tau^R\right\}.
 \end{equation*} 
  
 \item  Update the position and velocity functions $(X_t,V_t)$ for $t\in\left[t_{curr}, t_{curr}+\tau\right]$ following the flow solutions with the constants $C_1$ and $C_2$ given in Step \ref{item:tmvg:constants} above. Increase the current time by $\tau$: $t_{curr}\leftarrow t_{curr}+\tau$.
 
 \item Update the current velocity depending whether the bouncing, reflecting or refreshing event happened:
 \begin{enumerate}
 \item (bouncing) for $\tau=\tau^B$, $V_{t_{curr}}$ gets updated via kernel in \eqref{eq:kernel:R};
 \item (reflecting) for $\tau=\tau^{BB}$, the velocity gets updated according to \eqref{eq:tmvg:reflected:velocity}. 
 \item (refreshing) for $\tau=\tau^R$, $V_{t_{curr}}\sim\mathcal{N}(0,I_d)$.
 \end{enumerate}

\end{enumerate}

\subsection{Simulation results}

We compare the Quadratic Bouncy Hybrid Sampler to the Gibbs sampler in an example of a truncated bivariate normal taken from \cite{pakman2014exact}. The distribution is $\begin{pmatrix} x_1\\ x_2\end{pmatrix}\sim\mathcal{N}\left(\begin{pmatrix} 4 \\ 4\end{pmatrix}, I_4\right)$ truncated to $x_1\leq x_2\leq 1.1x_1$ and $x_1,x_2\geq 0$. The initial point for both samplers is $(1, 1.1)$. We run both samplers for the total of 100 times, where each chain is given the CPU time of 3 seconds. Note that since the Gibbs sampler is faster than than our proposed sampler (at least in our current implementation), for a given time the Gibbs sampler produces 12,000-15,000 samples while the QBHS produces 2,000-3,000.
 Denoting the true marginal means and variances of $x_1$ and $x_2$ as $(\mu_1, \sigma_1^2)$ and $(\mu_2,\sigma^2_2)$ and their estimates from a chain as $(\hat{\mu}_1, \hat{\sigma}_1^2)$ and $(\hat{\mu}_2,\hat{\sigma}^2)$, respectively, we compute mean squared error of these estimates (Table \ref{table:bivariate:normal}). The true values are computed by numerical integration. Even though the QBHS produces much less samples than the Gibbs sampler in a given time, we can see from the results that the QBHS mixes much faster. This is further illustrated by the histograms in Figure \ref{hist:qhbs} and Figure \ref{hist:gibbs}, where we plot marginal distributions of $x_1$ and $x_2$ given the samples of a single chain.

\begin{center}
\begin{table}[h!]
\renewcommand{\arraystretch}{2.0}
\begin{tabular}{ |c|c|c| } 
\hline
 	  & Gibbs & QBHS \\
\hline
MSE$(\hat{\mu}_1,\mu_1)$  & 0.004208 & 0.002358   \\ 
MSE$(\hat{\mu}_2, \mu_2)$ &  0.004619 & 0.002548  \\ 
MSE$(\hat{\sigma}^2, \sigma_1^2)$ & 0.008184 & 0.001541 \\ 
MSE$(\hat{\sigma}^2, \sigma_2^2)$ & 0.009778 & 0.001883\\ 
\hline
\end{tabular}
\caption{Comparing the Gibbs sampler and the QBHS in the truncated bivariate normal example. The table shows that the QBHS estimates better the marginal means and variances of individual coordinates than the Gibbs sampler.}
\label{table:bivariate:normal}
\end{table}
\end{center}

\begin{figure}[h!]
	\includegraphics[width=8cm]{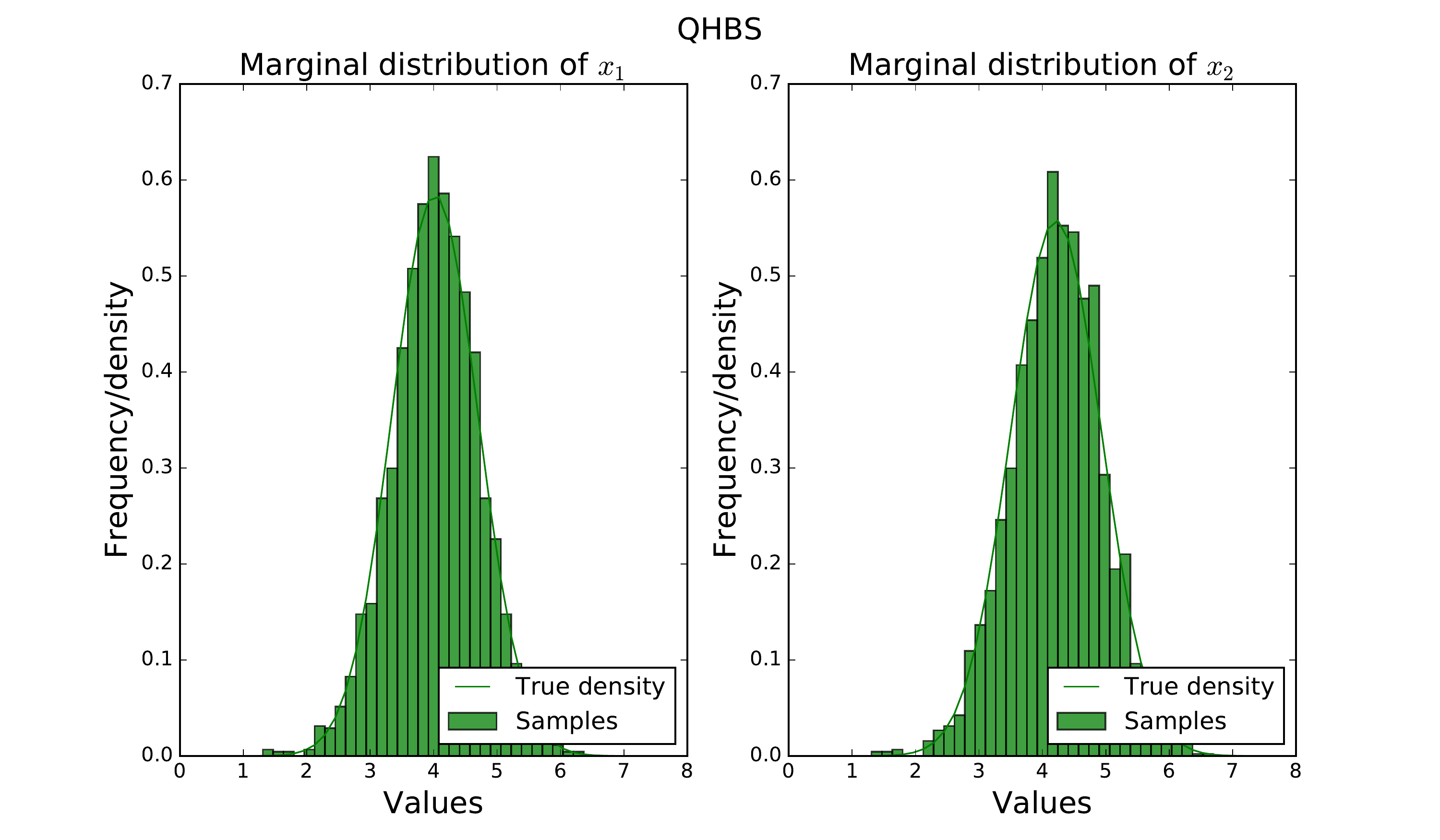}
	\caption{Histogram of marginal $x_1$ and $x_2$ samples by using the Quadratic Hybrid Bouncy Sampler.}
	\label{hist:qhbs}
\end{figure}

\begin{figure}[h!]
	\includegraphics[width=8cm]{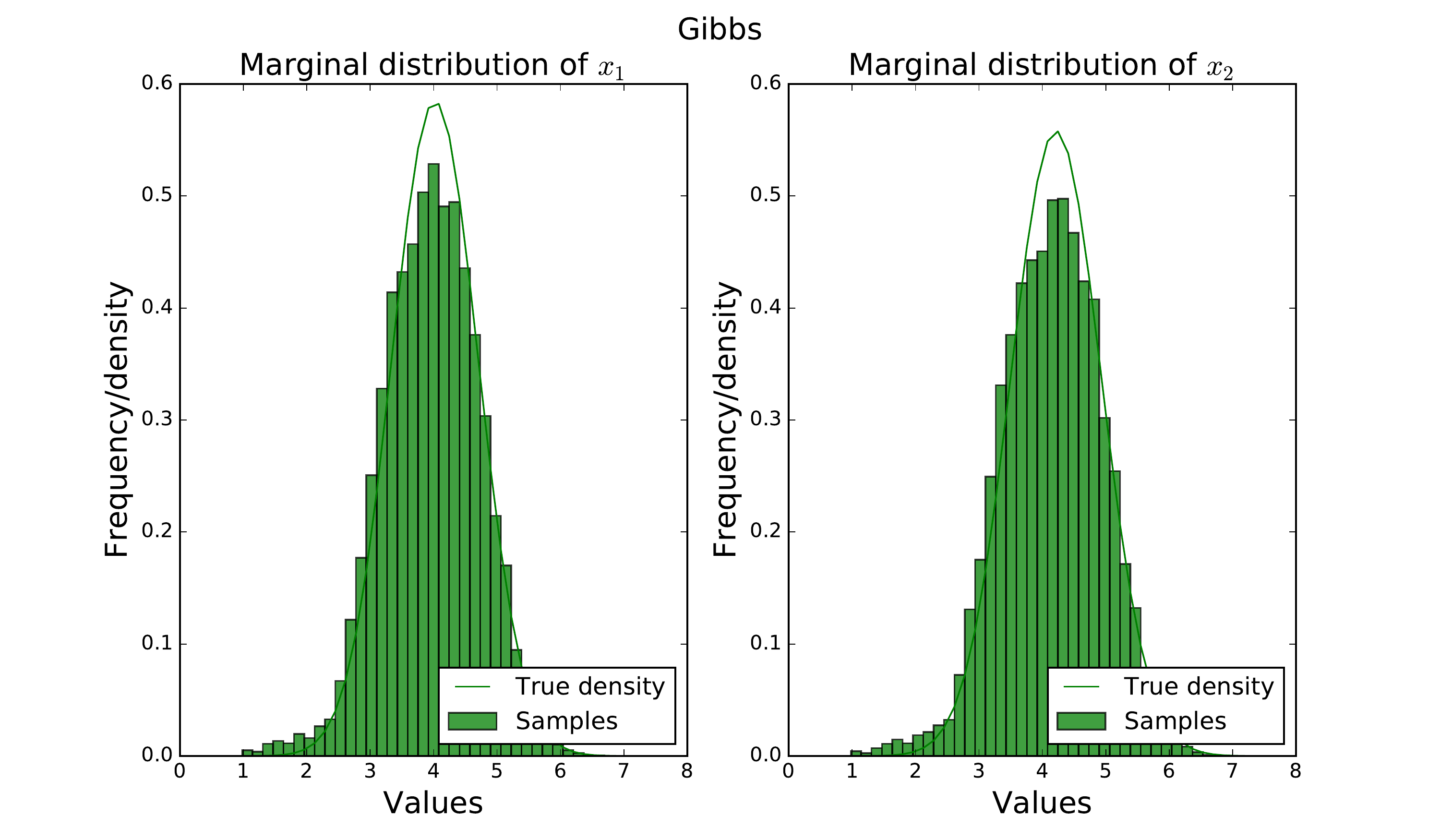}
	\caption{Histogram of marginal $x_1$ and $x_2$ samples by using the Gibbs sampler.}
	\label{hist:gibbs}
\end{figure}


\section{Further Extensions of Bouncy Hybrid Sampler} \label{sec:extensions}

\subsection{Modifying Bouncing Velocity} \label{sec:general:BPS}

We modify the BHS to have stochastic transition dynamics at the bouncy times rather than deterministic update using the kernel in \eqref{eq:kernel:R}.
Our modified sampler, called \textit{Stochastic Bouncy Hybrid Sampler}, includes as a special case the Generalized Bouncy Particle Sampler (GBPS) introduced in \cite{wu2017generalized}. 

At a bouncy time, the velocity $v$ has a component $v_p$ parallel to $g(x)$ and a component $v_o$ orthogonal to this direction, where $x$ is the current position of the particle. The transition dynamics flips the parallel sub-vector $v_p$ and resamples the orthogonal sub-vector with respect to some distribution.
The new transition kernel becomes
\begin{equation} \label{eq:generalized:kernel}
	Q\left(dx',dv'\mid x,v\right)=\delta_{x'}(x)\delta_{-v_p}\left(dv_p'\right)\mathcal{N}_{v_p^{\perp}}\left(dv_o'\right),
\end{equation}
where 
\begin{equation*}
\begin{aligned}
	v_p &= \frac{\langle v, g(x)\rangle}{\langle g(x), g(x)\rangle}g(x), \qquad v_o=v-v_p, \\
	v_p' &= \frac{\langle v',g(x)\rangle}{\langle g(x), g(x)\rangle}g(x), \qquad v_o'=v'-v_p',
\end{aligned}
\end{equation*}
and $\mathcal{N}_{v_p^{\perp}}$ is the $(d-1)$-dimensional standard normal distribution over the space $v_p^{\perp}$. Taking $g(x)=\nabla U(x)$ in the Bouncy Hybrid Sampler from Section \ref{sec:hmc} with the kernel in \eqref{eq:generalized:kernel} gives GBPS.

The infinitesimal generator of the process becomes
\begin{equation*}
\begin{aligned}
	\mathcal{A}f=&\langle\nabla_x f, v\rangle+\langle \nabla_v f,-\nabla U(x)+g(x)\rangle-\lambda(x,v)f(x,v)\\
	&+\lambda(x,v)\int_{v'\in\mathbb{R}^d} f(x,v')Q(dv'\mid x,v).
\end{aligned}
\end{equation*}
The following proposition proves the target distribution is invariant with respect to the new process for a general function $g$.

\begin{prop}\textnormal{\textbf{[Invariant density of Stochastic BHS]}} \label{thm:invariance:generalized} Assuming $U$ is continuously differentiable and $g$ is integrable,
the above piecewise deterministic Markov chain admits $\rho(x,v)=\pi(x)\psi_d(v)$ as its invariant distribution.
\end{prop}

\begin{proof}
	From the proof of Theorem \ref{prop:invariance}, we have
	\begin{equation*}
	\begin{aligned}
	&\int\limits_{x\in\mathbb{R}^d}\int\limits_{v\in\mathbb{R}^d} \bigl(\langle\nabla_x f, v\rangle	+\langle\nabla_v f,-\nabla U(x)+g(x)\rangle\bigr)d\rho(x,v)  \\
	&=   \int\limits_{x\in\mathbb{R}^d}\int\limits_{v\in\mathbb{R}^d} \langle v,g(x)\rangle f(x,v)d\rho(x,v)
	\end{aligned}
	\end{equation*}
	for any $f$ satisfying some regularity conditions, e.g.~boundedness and differentiability.
	Similarly to the proof of Theorem 1 in \cite{wu2017generalized}, we assume without the loss of generality that $v=(v_1,\ldots,v_d)$ decomposes into the sum of $v_p$ and $v_o$ with $v_p=(v_1,0, \ldots, 0)$, implying 
	\begin{equation*}
	\max\left\{0, \langle v, g(x)\rangle\right\}=\max\left\{0,\langle v_p, g(x)\right\}=\max\left\{0,\langle(v_1,0,\ldots,0), g(x)\rangle\right\},
	\end{equation*}
	where we used the fact that $v_p$ is parallel to and $v_o$ is orthogonal to $g(x)$.
	This implies
	\begin{equation}
	\begin{aligned}
	&\int\limits_{v\in\mathbb{R}^d}\int\limits_{v'\in\mathbb{R}^d}f\left(x,v'\right)(v^\top g(x))_+\pi(x)\psi_d(v)Q(dv'|x,v)dv\\
	&=\int\limits_{v_1\in\mathbb{R}}\int\limits_{(v_2,\ldots, v_d)\in\mathbb{R}^{d-1}}\int\limits_{v_1'\in\mathbb{R}}\int\limits_{(v_2',\ldots,v_d') \in\mathbb{R}^{d-1}}
	f\left(x,v_1',\ldots, v_d'\right)\left(v^\top g(x)\right)_+ \\
	&\qquad\cdot\pi(x)\psi_1(v_1)\psi_{d-1}\left(v_2,\ldots, v_d\right)\delta_{-v_1}(v_1')\psi_{d-1}(v_2',\ldots, v_d')dvdv'\\
	&=\int\limits_{v_1\in\mathbb{R}}\int\limits_{(v_2,\ldots, v_d)\in\mathbb{R}^{d-1}}\int\limits_{(v_2',\ldots,v_d') \in\mathbb{R}^{d-1}}
	f\left(x,-v_1,v_2', \ldots, v_d'\right)\left((v_1,0,\ldots,0)^\top g(x)\right)_+ \\
	&\qquad\cdot\pi(x)\psi_1(v_1)\psi_{d-1}\left(v_2,\ldots, v_d\right)\psi_{d-1}(v_2',\ldots, v_d')dv_1\ldots dv_d dv'_2\ldots dv'_d.
	\end{aligned}
	\end{equation}
	By the change of variables $v_1\rightarrow -v_1$, we have
	\begin{equation*}
	\begin{aligned}
	&\int\limits_{v_1\in\mathbb{R}}\int\limits_{(v_2,\ldots, v_d)\in\mathbb{R}^{d-1}}\int\limits_{(v_2',\ldots,v_d') \in\mathbb{R}^{d-1}}
	f\left(x,v_1,v_2',\ldots,v_d'\right)\left(-(v_1,0,\ldots,0)^\top g(x)\right)_+ \\
	&\qquad\cdot\pi(x)\psi_1(v_1)\psi_{d-1}(v_2,\ldots, v_d)\psi_{d-1}\left(v_2',\ldots, v_d'\right)dv_1\ldots dv_d dv'_2\ldots dv'_d.
	\end{aligned}
	\end{equation*}
	Integrating out variables $(v_2,\ldots, v_d)$, we have
	\begin{equation*}
	\begin{aligned}
	&\int\limits_{v_1\in\mathbb{R}}\int\limits_{(v_2',\ldots,v_d') \in\mathbb{R}^{d-1}}
	f\left(x,v_1,v_2', \ldots, v_d'\right)\left(-(v_1,0,\ldots,0)^\top g(x)\right)_+ \\
	&\qquad\cdot\pi(x)\psi_1(v_1)\psi_{d-1}\left(v_2',\ldots, v_d'\right)dv_1dv'_2\ldots dv'_d \\
	&=\int\limits_{v\in\mathbb{R}^d}f(x,v)\left(-v^\top g(x)\right)_+\pi(x)\psi_d(v)dv.
	\end{aligned}
	\end{equation*}

	Combining the above two results we have
	\begin{equation*}
	\begin{aligned}
	&\int\limits_{x\in\mathbb{R}^d}\int\limits_{v\in\mathbb{R}^d}	 \mathcal{A}f(x,v)\pi(x)\psi_d(v)dxdv
	=\int\limits_{x\in\mathbb{R}^d}\int\limits_{v\in\mathbb{R}^d} v^\top g(x)f(x,v)\pi(x)\psi_d(v)dxdv\\
	&\qquad +\int\limits_{x\in\mathbb{R}^d}\int\limits_{v\in\mathbb{R}^d}\left(\left(-v^\top g(x)\right)_+-\left(v^\top g(x)\right)_+)\right)f(x,v)\pi(x)\psi_d(v)dxdv.
	\end{aligned}
	\end{equation*}
	Since $\left(-v^\top g(x)\right)_+-\left(v^\top g(x)\right)_+=-v^\top g(x)$, the above expression equals zero.
\end{proof}

\subsection{Generalizing Refreshment Velocity} \label{sec:RHMC:velocity}

We describe a generalization of the refreshment velocity update in the Bouncy Hybrid Sampler from Section \ref{sec:hmc}. The derived family then includes the most general version of Randomized Hamiltonian MCMC from \cite{bou2017randomized} as a special case.

At the refreshment time, we update the velocity from $v$ to $v'$ as 
\begin{equation*}
	v'=\cos(\phi)v+\sin(\phi)\xi,
\end{equation*}
where $\xi\sim\mathcal{N}(0,I_d)$ and $\phi\in(0,\pi/2]$ is a deterministic parameter. The parameter $\phi$ governs how much of the updated velocity depends on the velocity prior to jump.
$\phi=\pi/2$ corresponds to the BHS and in this case $v'$ does not depend on $v$ and is completely random. The infinitesimal generator for the derived sampler for general $\phi$ becomes
\begin{equation*}
\begin{aligned}
	\mathcal{A}f&=\left\langle\nabla_x f,v\right\rangle+\left\langle\nabla_v f,-\nabla U(x)+g(x)\right\rangle-\lambda(x,v)f(x,v)\\
	&+\max\left\{0,\langle v,g(x)\rangle\right\}f(x,R(x)v)+\lambda_0\int\limits_{\xi\in\mathbb{R}^d}f\bigl(x,\cos(\phi)v+\sin(\phi)\xi\bigr)\psi_d(d\xi).
\end{aligned}
\end{equation*}
Using Proposition \ref{prop:InvarianceThm} of this paper and Proposition 3.1 of \cite{bou2017randomized}, we have that the target distribution is invariant for the proposed chain.
For $g(x)=0$, this infinitesimal operator becomes exactly equal to the infinitesimal operator in \cite{bou2017randomized}.


\section{A Family of Coordinate Hybrid Monte Carlo Samplers} \label{sec:zig:zag}

This section introduces a novel class of samplers called the Coordinate Bouncy Hybrid Samplers (CBHS) as yet another application of the piecewise deterministic Markov process framework. It is an infinite class of samplers, whose the velocity update happens only along one coordinate. At every coordinate switch the update only changes a single coordinate of the velocity. This sampler generalizes the Zig-Zag process of \cite{bierkens2016zig}.

Each function $g(x):\mathbb{R}^d\rightarrow\mathbb{R}^d$ specifies another sampler in this class. Denote the coordinates of such function $g$ as $g(x)=\left(g_1(x),\ldots, g_d(x)\right)$. Let the function $\gamma_i(x,v):\mathbb{R}^d\times\mathbb{R}^d\rightarrow\mathbb{R}$ be such that $\gamma_i(x,v)=\gamma_i(x,R_iv)$, where $R_i$ is the identity matrix with the element at $(i,i)$ set at -1.
 Denote the total time as $T_{total}$ and the current time as $t_{curr}$, where initially $t_{curr}\leftarrow 0$. This family of samplers can be described by performing the following steps while $t_{curr}<T_{total}$. 
\begin{enumerate}
	\item For each coordinate $i=1,\ldots, d$, solve the following system of differential equations for  position and velocity
	\begin{equation*}
	\begin{aligned}
		\dot{x} &=v \\
		\dot{v}_j&=\begin{cases} 0 & \textnormal{ for } j\neq i \\
		-\partial_{x_i} U(x)+g_i(x) & \textnormal{ for }j=i \end{cases}
	\end{aligned}
	\end{equation*}
	 such that the initial state of the solution at $t=0$ is at the the current state $(X_{t_{curr}},V_{t_{curr}})$. The solution becomes $(i, x_t^i,v_t^i)$, $i=1,\ldots,d$.
	
	\item For each $i=1,\ldots, d$, set the Poisson rate
	\begin{equation*}
		\lambda_i\left(x_t^i,v_t^i\right)=\max\left\{0, v^i_{i,t}\cdot g_i(x_t^i)\right\}+\gamma_i\left(x_t^i,v_t^i\right),
	\end{equation*}
	where $v_{i, t}^i$ denotes the $i$-th coordinate of $v_t^i.$
	
	\item For each $i=1,\ldots, d$, sample $\tau_i$ as the first arrival time of a Poisson process with the rate given above, i.e.~from a distribution 
	\begin{equation*}
		\mathbb{P}\{\tau_i\geq t\}=\exp\left(-\int_0^t\lambda_i(x_t^i,v_t^i)ds\right).
	\end{equation*}
 	 Set the moving time $\tau$ to be the smallest of the sampled times 
 	\begin{equation*}
 		\tau=\underset{1\leq i\leq d}\min\tau_i.
 	\end{equation*} 
 	Denote $i_0=\arg\underset{1\leq i\leq d}\min \tau_i$, so that $\tau=\tau_{i_0}$.
 	
	\item The particle moves along $\left(x_t^{i_0},v_t^{i_0},i_0\right)$ for time $\tau=\tau_{i_0}$, i.e.~$\left(X_{t_{curr}+t},V_{t_{curr}+t}\right)=\left(x_t^{i_0},v_t^{i_0}\right)$, $t\in[0,\tau]$. Update the current time $t_{curr}\leftarrow t_{curr}+\tau$.
	
	\item Update the velocity at the current time according to the matrix $R_{i_0}$ as follows
	\begin{equation*}
			R_{i_0}v = \begin{cases}
				v_i & \textnormal{ for } i\neq {i_0}\\
				-v_i & \textnormal{ for } i=i_0,
			\end{cases}
	\end{equation*}
	switching only the $i_0$-th coordinate of the current velocity.
\end{enumerate}

The infinitesimal generator of the process above equals $\mathcal{A}f=\sum_{i=1}^n\mathcal{A}_if$, where
\begin{equation*}
	\mathcal{A}_if = \partial_{x_i} f\cdot v_i+\partial_{v_i} f\cdot\bigl(-\partial_{x_i} U(x)+g_i(x)\bigr)+\lambda_i(x,v)\bigl(f(x,R_iv)-f(x,v)\bigr).
\end{equation*}

\begin{prop}\textnormal{\textbf{[Invariant density of CBHS]}}
	Assuming $U$ is continuously differentiable and $g$ is integrable, the measure $\pi(x)\bar{\psi}_d(v)$ is a stationary measure for the process above, where $\bar{\psi}_d(\cdot)$ is any density invariant under $R_i$ for all $i=1,\ldots, d$, i.e.~$\bar{\psi}_d(v)=\bar{\psi}_d(R_iv)$ for all $v\in\mathbb{R}^d$ and all $i=1,\ldots,d$.
\end{prop}

\begin{proof}
	Using integration by parts
	\begin{equation*}
	\begin{aligned}
	&\int_{\mathbb{R}^d}(\partial_{x_i}f\cdot v_i)\pi(x)dx=\int_{\mathbb{R}^d} (v_i\cdot \partial_{x_i}U(x))f(x,v)\pi(x)dx \;\;\textnormal{ and }\\
	&\int_{\mathbb{R}^d}\partial_{v_i}f\cdot\bigl(-\partial_{x_i}U(x)+g_i(x)\bigr)\pi(x)dx=\int_{\mathbb{R}^d} \bigl(v_i\cdot (-\partial_{x_i}U(x)+g_i)\bigr)f(x,v)\pi(x)dx
	\end{aligned}
	\end{equation*}
	Thus for the first two terms we have
	\begin{equation*}
		\int_{\mathbb{R}^d}\bigl(\partial_{x_i}f\cdot v_i
		+\partial_{v_i}f\cdot(-\partial_{x_i}U(x)+g_i(x)\bigr)\pi(x)dx=\int_{\mathbb{R}^d}(v_i\cdot g_i)f(x,v)\pi(x)dx.
	\end{equation*}
	Using the change of variables $u=R_iv$, we have 
	\begin{equation*}
	\begin{aligned}
	&\int_{(x,v)\in\mathbb{R}^{2d}}\bigl((v_i\cdot g_i(x))_++\gamma_i(x,v)\bigr)f(x,R_iv)\pi(x)\psi_d(v)dxdv \\
	&=\int_{(x,u)\in\mathbb{R}^{2d}} \bigl((-u_i\cdot g_i(x))_++\gamma_i(x,R_iu)\bigr)f(x,u)\pi(x)\psi_d(u)dxdu,
	\end{aligned}
	\end{equation*}
	implying
	\begin{equation*}
	\begin{aligned}
	&\int_{(x,v)\in\mathbb{R}^{2d}}\lambda_i(x,v)\left(f(x,R_iv)-f(x,v)\right)\pi(x)\psi_d(v)dxdv \\
	&= \int_{(x,v)\in\mathbb{R}^{2d}} \left((-v_i\cdot g_i(x))_+-(v_i\cdot g_i(x))_+\right)f(x,v)\pi(x)\psi_d(v)dxdv\\
	&=-\int_{(x,v)\in\mathbb{R}^{2d}} \left(v_i\cdot g_i(x)\right)f(x,v)\pi(x)\psi_d(v)dxdv.
	\end{aligned}
	\end{equation*}
	It follows $\int_{(x,v)\in\mathbb{R}^{2d}}\mathcal{A}_if(x,v)\pi(x)\psi_d(v)dxdv=0$, proving the proposition.
\end{proof}

\begin{remarks}
$\:$
\begin{itemize}[leftmargin=*]
\item Taking $g(x)=\nabla U(x)$, the sampler above becomes the Zig-Zag process of \cite{bierkens2016zig}.  
\item Note that not all densities $\bar{\psi}_d$ invariant under $R_i$ for all $i=1,\ldots,d$, will be irreducible for the described chain. 
\end{itemize}
\end{remarks}

\section{Conclusion}

This work introduces an infinite class of samplers, generalizing and connecting together the recent samplers, including the Bouncy Particle Sampler, Hamiltonian Markov chain and the Zig-Zag process. Our sampler is a piecewise deterministic Markov processes, whose trajectories are not necessarily linear but governed by the solution of the system of differential equations. The moving time along each of these trajectories is simulated as the first arrival time of a  corresponding Poisson process.
We proved the proposed sampler has the target distribution as its invariant/stationary distribution. 
There are already results showing some specific instances of our sampler are ergodic, including the Randomized Hamiltonian MC, the BPS and the Zig-Zag. The conditions under which the general proposed process is ergodic, including a rate of convergence, are left for future work.

A related question involves investigating the distribution-dependent choice of function $g$. We suspect that the mixing time of the proposed sampler will depend on $g$ with the optimal $g$ being different for different distribution. 

\section*{Acknowledgments} 

The authors would like to thank Persi Diaconis and Jonathan Taylor for helpful discussions.


\bibliography{PDMP}
\bibliographystyle{plain}

\end{document}